\documentclass[a4paper, 12pt]{article}
\usepackage{amssymb, graphicx, enumerate, color}
\usepackage[authoryear]{natbib}
\usepackage[ruled]{algorithm2e}
\usepackage{amssymb, graphicx, enumerate, color}
\usepackage{graphicx}
\usepackage{amsmath}
\usepackage{amssymb}
\usepackage{amsthm}
\usepackage{xcolor}
\usepackage{hyperref}

\hypersetup{
    colorlinks,
    linkcolor={red!50!black},
    citecolor={blue!50!black},
    urlcolor={blue!80!black}
}

\usepackage{algorithmic}

\newtheorem{theorem}{Theorem}[section]

\newtheorem{lemma}[theorem]{Lemma}

\renewcommand{\leq}{\leqslant}
\renewcommand{\geq}{\geqslant}

\SetAlFnt{\small}
\SetAlCapFnt{\small}
\SetAlCapNameFnt{\small}
\SetAlCapHSkip{0pt}
\IncMargin{-\parindent}

\title{Pricing Policies for Selling Indivisible Storable Goods\\ to Strategic Consumers\\[-1cm]}
\begin{document}
\bibliographystyle{plainnat}

\date{\today}


\maketitle

\renewcommand*{\thefootnote}{\fnsymbol{footnote}}

\vspace{-1cm}
\begin{center}
{\sc Gerardo Berbeglia}\footnote{Melbourne Business School, University of Melbourne. Email: {\tt g.berbeglia@mbs.edu} },
{\sc Gautam Rayaprolu}\footnote{McGill University. Email: {\tt gautam.rayaprolu@mail.mcgill.ca} }
and
{\sc Adrian Vetta}\footnote{McGill University. Email: {\tt adrian.vetta@mcgill.ca} }
\end{center}

\renewcommand*{\thefootnote}{\arabic{footnote}}
\addtocounter{footnote}{-3}

\begin{abstract}
We study the dynamic pricing problem faced by a monopolistic retailer who sells a storable product to forward-looking consumers. In this framework, the two major pricing policies (or mechanisms) studied in the literature are the preannounced (commitment) pricing policy and the contingent (threat or history dependent) pricing policy. We analyse and compare these pricing policies in the setting where the good can be purchased along a finite time horizon in indivisible atomic quantities. First, we show that, given linear storage costs, the retailer can compute an optimal preannounced pricing policy in polynomial time by solving a dynamic program. Moreover, under such a policy, we show that consumers do not need to store units in order to anticipate price rises. Second, under the contingent pricing policy rather than the preannounced pricing mechanism, (i) prices could be lower, (ii) retailer revenues could be higher, and (iii) consumer surplus could be higher. This result is surprising, in that these three facts are in complete contrast to the case of a retailer selling divisible storable goods \citep{dudine2006storable}. Third, we quantify exactly how much more profitable a contingent policy could be with respect to a preannounced policy. Specifically, for a market with $N$ consumers, a contingent policy can produce a multiplicative factor of $\Omega(\log N)$ more revenues than a preannounced policy, and this bound is tight.

\end{abstract}

\section{Introduction}
The problem of determining an optimal pricing policy for a monopolistic retailer has attracted the attention of many researchers from Economics and Operations Research. These problems are generally studied in a dynamic game setting in which the players are the single retailer and the consumers. At each period, the seller sets a selling price (under pre-announced pricing all prices are set in advance) and buyers decide whether to buy or wait (and potentially buy a product later). This interplay between a retailer and consumers may sometimes lead to surprising equilibria. Some relevant questions in this topic are: What does the price path look like in an equilibrium? How profitable can the retailer be, and what pricing policy should it use? Is it possible to compute the optimal pricing policy efficiently?
In this paper we study the dynamic pricing problem faced by a retailer (a monopolist) who sells storable goods to forward-looking consumers along a finite time horizon. Storable goods are goods that can be stored for a long period but then loses all value after its
first use. Prominent examples of storable goods include non-perishable grocery items and many natural commodities, for example, oil and natural gas. A less conspicuous example concerns digital goods. Digital goods are durable, but the technology often exists to make them perishable and, thus, convert them into storable goods. \footnote{A firm may wish to reduce
the lifespan of a durable good to avoid a Coasian outcome, that is, one where the firm
loses all monopoly power. Indeed, Amazon has recently been granted a US patent (8,364,595)
for a second-hand market in digital goods that restricts the number of times a good can be resold, that is, limits
the longevity of the goods. Apple has applied for a similar patent (US 20130060616 A1).}

Dynamic pricing, a practice which is very common in most industries, can be seen as a weak form of price discrimination. The retailer does not dictate how consumers are segmented, but rather consumers segment themselves by self-selecting the time (and corresponding price) at which they prefer to buy the item. A classic and well studied dynamic pricing problem occurs when a monopolist sells a durable good to strategic consumers.\footnote{A {\em durable good} is a long-lasting good that can be ``stored" and consumed repeatedly over time.
Examples include digital goods, land and, to a large degree, housing, metals, electronic goods, etc.}. A natural question in this setting is whether the seller is better-off by pre-announcing all future prices (i.e. price-commitment) rather than using a price-contingent strategy (i.e. revealing prices dynamically). Ronald \citet{coase1972durability} conjectured that a durable good monopolist who lacks price commitment
cannot sell the good above the competitive price. \citet{stokey1979intertemporal} then showed that the optimal preannounced pricing policy consists of setting a fixed price. An intuitive argument of why a preannounced pricing policy might be beneficial to the retailer is the following: in an environment with dynamic pricing, consumers may sometimes hesitate whether they should buy a product today or wait for a potential discount in the future. In some contexts, this may hurt the seller's profit. On the other hand, a pre-announced pricing policy gives a seller the power of commitment and thus, can induce consumers to buy at a high price since they can be certain that prices will not drop enough to make it beneficial to wait.

The Coase conjecture was first proven by \citet{gul1986foundations} under an infinite time horizon model with non-atomic consumers. However, the conjecture relied upon several assumptions, such as that the time between periods approaches zero and that there is an continuum of consumers. Since then, there have been several studies regarding the performance of different pricing policies for various economic models for durable goods, perishable goods, under different time horizons. \citet{guth1998durable} studied a finite horizon model with a continuum of consumers (i.e. non-atomic demand) and showed that when consumer valuations follow a uniform distribution, there exists a subgame perfect equilibrium, as period lengths approach zero, in which the monopolist profits converge to the static monopoly profits discounted to the beginning of the game. \citet{bagnoli1989durable} studied the durable good monopoly problem in which there are finite (atomic) consumers (i.e. discrete demand). When the time horizon is infinite, they proved a surprising result: the
existence of a subgame perfect Nash equilibrium in which the duropolist extracts all the consumer surplus. This equilibrium refutes the Coase conjecture when there is a finite set of atomic buyers and an infinite time horizon. Indeed, it shows there exists a subgame perfect Nash equilibrium where duropoly profits exceed static monopoly profits by an unbounded factor. \citet{berbeglia2014} studied the same model as \citet{bagnoli1989durable} and proved that when there is a finite time horizon the monopolist profits cannot exceed twice the static monopoly profits. A major difference between the durable good literature (in particular \citet{bagnoli1989durable} and \citet{berbeglia2014}) and our paper is that in the durable good monopoly problems, consumers have a \emph{single} demand and the valuation is assigned for the \emph{whole game}. In our model, however, consumers can have a demand of multiple units and the valuation obtained obtained is time-dependent. We also showed in Section 4.2 that the duropoly model of \citet{berbeglia2014} can be seen as a special case of the storable good model presented in this paper. Nevertheless, we showed that the conclusions with regards to the benefits of price commitment obtained in \citet{berbeglia2014} do not generalize to this broader class of models.

\citet{dasu2010dynamic} analyzed preannounced pricing policies for a retailer for the case in which consumers are finite and the number of items are limited (i.e. there is an availability constraint). The authors showed that dynamic pricing (through preannounced pricing) is beneficial for the retailer, but the benefit decreases as the market size increases.
\citet{su2007intertemporal} studied contingent (subgame perfect or threat) pricing  policies in a finite time horizon model where consumers, who arrive continuously, are heterogeneous in their valuation and in their patience. The author characterized the structure of optimal pricing policies and showed that, depending on what population type dominates the market, prices are increasing or decreasing.  \citet{aviv2008optimal} studied contingent pricing and preannounced pricing equilibria for a related two-period model of a retailer selling a finite number of items to strategic consumers that arrive under a Poison process. The authors proved that a preannounced policy, with respect to contingent pricing, can increase monopoly profits by up to eight percent. \citet{li2001pricing} consider optimal pricing policies for another model where items are also limited inspired by the revenue management problem faced by airlines. The potential benefit of strategic capacity rationing under preannounced pricing was analyzed by \citet{ liu2008strategic} in a two period pricing model of a durable good. Their main result is that rationing can be beneficial for the retailer but only when consumers are risk-adverse.  \citet{correa2016contingent} studied a new type of pricing mechanism so called {\em contingent preannounced pricing} in which the retailer commits to a price menu which states what would be the future price as a function of the number of available inventory. The authors showed that under this pricing policy, there exists a unique equilibrium when there is a single unit of inventory, but multiple equilibria may exist otherwise. Computational experiments showed that the contingent preannounced pricing outperforms previously proposed pricing mechanisms. Very recently, \citet{surasvadi2017using} examined a contingent price markdown mechanism with guaranteed reservation in which the retailer sells a limited number of items over a finite time horizon. The authors have identified market conditions in which the proposed mechanism outperforms a price commitment policy. Other models of durable good markets that incorporate multi-generation products and the effects of price sensitivity has been studied by \citet{kuo2012dynamic} and \citet{liberali2011effects} respectively.


A related setting in which there is a complex dynamic interaction between a retailer and the consumers is when the good is storable -- a {\em storable good} is a good that can be
stored prior to consumption but it is not durable as it perishes upon consumption.
Consequently, as stated, such a good can be stored for a long period but then loses all value after its
first use.
Prominent examples of storable goods include non-perishable grocery items and many
natural commodities, for example, oil and natural gas. A less conspicuous example concerns digital goods.
As with durable good markets, the consumers in a storable good market can strategically time their purchases by taking into
account current and expected future prices. In a storable good market, a consumer who observes a low price
promotion sale, may increase her purchases not only to augment the current consumption
level (the {\em consumption effect}) but also to stockpile the good for future consumption (the {\em stockpiling effect}).
There is a large amount of evidence from empirical work showing that the stockpiling effect -- also known as {\em demand anticipation} --
is non-negligible in markets for storable goods.
For example, \citet{pesendorfer2002retail} investigated supermarket price reductions for ketchup
products in Springfield, Missouri. His main finding is that demand at
low prices depends on past prices, which suggests the existence of a
stockpiling effect. Similarly, \citet{hendel2006sales}
examined the price and demand dynamics of soft drinks, laundry detergents,
and yogurt from nine supermarkets in the US. This analysis also showed that demand
increases during price reductions are, in part, due the stockpiling effect.
To capture the potential for stockpiling, the models of storable good markets
assume that consumers have access to inventory space for an associated cost.\footnote{These storable goods
models are also referred to as {\em dynamic inventory models}.}
Again, the question of whether the seller is better-off by pre-announcing all future prices (i.e. price-commitment) rather than using a price-contingent strategy (i.e. revealing prices dynamically) is also important in this setting. While it is generally the case that firms do not pre-announce future prices in most markets, there are several examples where, directly or indirectly they provide an approximation of the pricing schedule. In clothes retailing, consider for example the American luxury department store Nordstrom. As stated in \citet{elmaghraby2009will}, Nordstrom pre-announces several weeks in advance, the starting dates of its July sales. While the exact price reduction is not announced exactly, consumers know that most items would be 25-50 percent off.
In the airline industry, some airlines provide mechanisms to encourage consumers not to strategically schedule their purchases. For example, South West pre-announced future prices for some flights \citep{heskett2010southwest}. Other airlines, such as easyJet, do not pre-announce explicitly their prices, but commit to a monotone increasing price policy \citep{koenigsberg2008easyjet}. In the theatre industry, \citet{tereyaugouglu2017pricing} empirically tested whether a performing art organization could increase sales by using a pre-announced pricing strategy. Their results suggest that for this organization, price-commitment is better than a contingent pricing policy. In other industries, pre-announced pricing mechanisms and dynamic pricing policies coexist, we provide two examples. The first one is the multi-billion dollar industry of cloud computing services such as Amazon web services. These platforms provide several long-term contract services where prices are pre-announced, together with short term contracts whose prices are contingent to the current demand (spot market pricing). The second example is transportation services such as those provided by Uber and Lyft, where prices are typically pre-announced, but the policy may switch to a dynamic pricing scheme in situations with high demand.


\citet{dudine2006storable} studied a market where the retailer sells a storable good to
consumers with time-dependent demand over an arbitrary (but finite) number of time periods. They
proved that, under some mild assumptions on consumer demand, if the retailer uses a preannounced price
mechanism then consumers will not stockpile any goods in equilibrium. They also showed that (i) prices are higher,
(ii) profits are lower, and (iii) consumer surplus is lower if the retailer uses a contingent pricing
mechanism rather than a preannounced price mechanism. These results also suggest, unlike the classical
results that proved the Coase conjecture
in durable good settings, that monopoly power in storable good markets leads to several inefficiencies. The results
of \citet{dudine2006storable} have later been extended to other settings of a storable good market, namely a two-level vertically
integrated monopoly \citep{nie2009commitment}, and a duopoly market model \citep{nie2009commitment_duopoly}.

A major assumption in the storable good market model of \citet{dudine2006storable} is that the good is infinitesimally divisible. In addition, they consider that either there is a single consumer in the market who can always obtain some positive additional utility by consuming an additional fraction of the good (no matter how small that fraction is), or that there is a continuum of non-atomic buyers. In our paper, we will relax those assumption by imposing that goods are not divisible. Other models for storable goods studied in the literature include a continuous time stochastic inventory model for a commodity traded in a spot market \citep{chiarolla2015optimal} and a model where the storable good has a fixed shelf-life \citep{herbon2014optimal}. The literature of lot-sizing is also related. Lot-sizing models are constructed to optimize production planning and inventory management also in a multi-period setting. An important difference, however, is that in lot-sizing problems consumers are not strategic and therefore the optimization problem is not subject to equilibrium constraints (see, e.g. \citet{drexl1997lot}, \citet{jans2008modeling}, and \citet{robinson2009coordinated}).

%
%
%
%


\subsection{Contributions} \label{sec:results}
In this paper we consider the model of a retailer selling storable goods used in \citet{dudine2006storable}, but in which the items are {\em indivisible}. Specifically, we assume that either (1) there is a finite (possibly very large) number of buyers with a unitary demand per period or (2) there is a single buyer with an arbitrary demand per period, but can only obtain value from an integral number of items.

Our main contributions are as follows. We first characterize the optimal preannounced pricing strategy for the retailer.
We then show that it can be computed efficiently and that, under this preannounced price mechanism,
the consumers will not stockpile the good. We next show that it is possible for (i) prices to be lower, (ii) retailer's profits to be higher
and (iii) the consumer surplus to be higher, if the retailer uses a contingent pricing mechanism rather than a preannounced price mechanism. These last three results are surprising, since they are in complete contrast
to the case of a retailer for divisible storable goods \citep{dudine2006storable}.

We then quantitatively compare the profitability of the preannounced price and contingent pricing mechanisms.
Specifically, for a market with $N$ consumers, a contingent pricing mechanism
can generate $\Omega(\log N)$ times more profits than an optimal preannounced price mechanism,
and this bound is tight. This is notable because previous models with a finite number of periods
have the property that the contingent pricing mechanism cannot significantly outperform the preannounced price mechanism.
In particular, we explain how our result contrasts with related results for a durable
good retailer. Finally, we examine the setting where inventory costs are concave. Interestingly, the consumers may now utilise storage when the retailer uses a preannounced price mechanism.

\section{The Model}\label{sec:model}
Consider a retailer facing consumers with demands for a consumable good over $T$
time periods. The retailer can produce units of the good at a unitary cost $u$ and wishes to select prices $\{p_1, p_2, \dots, p_T\}$
to produce sale quantities $\{q_1, q_2, \dots, q_T\}$ that maximize profits, $\sum_{t=1}^T p_t \cdot q_t$. We assume, without loss of generality, that $u=0$ and consequently, profit and revenue are interchangeable in this setting.
For a perishable good with one-period shelf-life, the corresponding optimization problem consists of optimizing each time period independently, as the goods bought or produced in one period cannot be consumed in the next one. But observe that as the goods shelf-life grows, the problem becomes more complicated as consumers can strategically buy items in one period only to consume them at a later date.
Storability offers the consumers the opportunity to purchase the good {\em before}
they actually wish to use it. Indeed, this option is desirable if the resultant price savings exceed
storage costs. In accounting for this possibility, the retailer's optimization problem
now becomes non-separable.


We consider two cases for this problem: the \emph{single-buyer} and the \emph{multi-buyer} settings. In the former, all the demand comes from a single buyer whose (time-depedent) value obtained depends on the quantity of indivisible items consumed. In the latter, there is a finite number of consumers, each having a single (indivisible) unitary demand per period. In previous studies, the results obtained from these two settings are essentially identical. However, for an indivisible good, these two settings are {\bf not equivalent} and may produce different outcomes. Specifically, in Section 4, we show an example where, in the case of a single buyer with multiple units of demand, the retailer can extract more profit than in the case where there are multiple buyers with unitary demand. Furthermore, in Section 5, we show that the two settings have different outcomes with respect to the use of storage under concave storage costs.

Our interest lies in analysing and comparing two standard pricing mechanisms: the \emph{preannounced price mechanism}
and the \emph{contingent pricing mechanism}. In the preannounced price mechanism, the retailer publicly
announces in advance (that is, at time $t=0$) a price schedule $\{p_1, p_2, \hdots, p_T\}$ where $p_t$
is the per unit price in period $t$. The consumers can then make their purchasing and storage decisions based upon
these prices. In the contingent pricing mechanism, the retailer announces prices sequentially over time.
Thus, the announced price $p_t=p(\mathcal{H},t)$ may depend not only on the current time period $t$, but
also on past history $\mathcal{H}$. Once the price at period $p_t$ is announced, the consumer (or consumers) must decide how
many units to purchase and store without knowledge of future prices  (the consumer also has the opportunity to purchase goods in the future, even in the case where she has already purchased goods earlier in the game).

For such a sequential game, the solutions we examine are pure subgame perfect Nash equilibria (SPNE). In particular,
this restricts our attention to {\em credible} contingent pricing mechanisms. It worth mentioning that the decision on which pricing mechanism to use may affect not only the retailer profits, but also the total purchases (i.e. observable demand), as well as the consumer surplus.

Before we proceed to explain precisely the two models (single-buyer and multiple-buyers), we enumerate several assumptions which are needed. (i) The retailer has enough capacity to satisfy all demand; (ii) The per unit cost for the retailer is constant over time and it is set to zero (this is without loss of generality, as one can always normalize it); (iii) The retailer does not incur in inventory cost (or equivalently, it can receive the items from the distributor in every period without an ordering cost); (iv) the retailer and the consumers know the distribution of the consumer valuations; (v) all consumers are known at time zero (no new consumers appear in the game); (vi) consumers do not incur in an ordering cost or delivery cost.

\subsection{The Single-Buyer Case.}
We begin with case where the retailer faces a unique buyer (monopsonist) for its good.
Let the value the buyer receives from consuming $x$ units of the
good at period $t$ be given by the function $U(x,t)$. Units are assumed to be indivisible, and
therefore the domain of the utility function $U(x,t)$ is $\mathbb{N} \times [T]$.
Assume the buyer faces a set of prices $\{p_1\hdots, p_T\}$ and that the
storage cost is $c\ge 0$ dollars per unit per time period.
The consumer problem then consists of choosing
the consumption quantities $\{x_1,x_2,\hdots,x_T\}$, the number of units purchased per period $\{q_1,\hdots, q_T\}$,
and the storage levels $\{S_1, S_2,\hdots, S_T\}$.
These are chosen to maximize her total utility
\begin{equation}
\sum_{t=1}^{T} \left( U(x_t,t) -p_t\cdot q_t - c \cdot S_t\right)
\end{equation}\label{buyers_problem}

subject to the storage constraints that $S_t = S_{t-1} + q_{t} - x_t$, for all $t=1,\hdots,T$ where $S_0=0$.

The retailer selects the prices $\{p_1\hdots,p_T\}$ in order to maximise the resulting profits $\sum_{t=1}^{T} p_t\cdot q_t$.
As stated, for the preannounced price mechanism the prices are chosen in advance and
for the contingent pricing mechanism the prices are chosen sequentially.
In analyzing these two mechanisms we make the following standard assumptions.
First we assume the buyer's valuation function $U(x,t)$ is non-decreasing in $x$.
Second, we assume the marginal valuation function is non-increasing in $x$.
Specifically, let $V(x,t) = U(x,t)-U(x-1,t)$ represent the marginal value obtained by consuming $x$ units of the good at
period $t$. Then we have that $V(x,t)$ is non-increasing in $x$. We also assume that there exists
$H \in \mathbb{N}$ such that $V(H,t)=0$ for all $t \in T$. No other restrictions are
imposed on the consumer utilities. For notational purposes, we set $V(0,t)=L$ for all $t\in[T]$ where $L$ is a huge number.

\subsection{The Multi-Buyer Case.}
Next consider the set-up for the multi-buyer case. We now have a set of $N$ consumers each seeking
to maximize their own utility.
Each consumer $i$ has a demand of at most one unit per period and the value she obtains if she
consumes the unit at period $t$ is $u_{i,t}$.
A consumer $i$ may, however, purchase multiple units $q_{i,t}$ of the indivisible good in any time period and
store them for future consumption. Again, there is a storage cost of $c$ per period and per unit stored.
This time the retailer chooses prices $\{p_1,\hdots,p_T\}$
to maximize $\sum_{t=1}^Tp_t \cdot q_t$, where $q_t=\sum_i q_{i,t}$ is the total number of items sold in period $t$.

For the preannounced price mechanism, the consumers make their purchasing/storage decisions based
upon the entire set of prices. For the contingent pricing mechanism, in each time period, the
consumers decide (simultaneously) their purchasing/storage quantities for that period, given the current price.
As we shall see, these decisions impose constraints on the retailer's optimization problem.
In particular, for the contingent pricing mechanism these constraints force solutions to correspond to
pure subgame perfect Nash equilibria (SPNE).

To simplify notation, let $\sigma_t: [N] \rightarrow [N]$
denote the permutation that ranks the $N$ consumers in decreasing values of their utilities in period $t$
and let $v_{i,t} = u_{\sigma^{-1}(i),t}$ for all $i \in [N]$ and all $t \in [T]$. Again, we set $v_{0,t}= L$ where $L$ is a huge number.

%



\subsection{An Example}\label{sec:bad-example}
In this section, we present a simple example to illustrate the proposed model and the concepts involved. Then, in Section \ref{section_cont_example}, we consider the continuous version of the same example (i.e. the case in which the goods are divisible) to provide an insight of why the indivisibility of the goods alters some of the results obtained by \citet{dudine2006storable}.

Consider a two-period game with a per unit storage cost of $c=1$ and with consumer valuations given
in Table \ref{ExThreat}. Here we assume that we are in the multi-buyer case with $N=2$; however, it is straightforward to check that our analysis also applies if
there is a single-buyer whose marginal utilities for each extra unit are described by the same
table.
\begin{table}[ht]
\centering
\caption{A small example}
\begin{tabular}{|c||c|c|}
\hline
\textbf{Consumer} & \textbf{Value at t=1}  &\textbf{Value at t=2}  \\
\hline
 \textbf{1}  & 17 & 15  \\
\hline
 \textbf{2} & 10 & 4  \\
\hline
\end{tabular}
\label{ExThreat}
\end{table}



We first calculate the optimal {\em preannounced pricing strategy} for this game. Let $p_1$ and $p_2$ denote
an optimal sequence of prices under the preannounced pricing policy for periods 1 and 2 respectively. We claim that no
consumer will ever buy a unit at period 1 and then store it for consumption in period 2. Otherwise
we must have $p_2 > p_1 + 1$ for it to be beneficial to buy early and pay the storage cost.
But this cannot be optimal for the retailer as it would then benefit from
setting $p_2=p_1+1$. It follows that the optimal choice for $p_1$
must either be $p_1=17$ in which case only consumer $1$ buys, or be $p_1=10$ in which case
both consumers buy. In the former case, the optimal choice for $p_2$ would be $p_2=15$ yielding a total
profit of $17 \cdot 1 + 15 \cdot 1 = 32$. In the latter case, the optimal period two
price $p_2$ cannot be higher than $11$, otherwise consumer $1$ would use storage.
Consequently,  the optimal price is $p_2 = 11$, which yields the retailer a profit of $10 \cdot 2 + 11  \cdot 1 = 31$.
Therefore the optimal preannounced price strategy yields a profit of $32$ and consists of
setting prices $p_1=17$ and $p_2=15$.

Next we calculate the optimal {\em contingent pricing mechanism}.
Because prices and actions are chosen sequentially we can compute the optimal contingent pricing mechanism by reasoning backwards from the final time period.
So consider time $t=2$.
In order to select $p_2$ the retailer will base its decision on the set of items
currently stored. If consumer $1$ has not stored an item from period $1$ to $2$
then the retailer will set $p_2=15$, regardless of the prior decisions made by consumer $2$.
Knowing this, consumer $1$ will purchase two items in period $1$ provided
$p_1+c \le 15$, that is $p_1\le 14$.
This observation leads to the following four possibilities. (i) The retailer
sets $p_1=17$. Consumer $1$ then purchases one item in each time period
giving a profit of $17 \cdot 1 + 15 \cdot 1 = 32$.
(ii) The retailer sets $p_1=14$ and consumer $1$ purchases two units.
Given this, the retailer then charges $p_2=4$ and consumer $2$
buys one unit in the last period. This gives a profit of $14 \cdot 2 + 4\cdot 1 = 32$.
(iii) The retailer sets $p_1=10$ and consumer $1$ purchases two units and consumer $2$
purchases one unit. Again, the retailer then charges $p_2=4$ and consumer $2$
buys one unit in the last period. This gives a profit of $10 \cdot 3 + 4\cdot 1 = 34$.
(iv) In order to induce consumer $2$ to purchase two units in period $1$ the
retailer will have to charge $p_1\le 3$ giving profits of at most $3\cdot 4=12$.
Note that option (iii) is a credible threat. If consumer $1$ refuses to buy two units in period
$1$ her utility will fall as the price would then rise to $p_2=15$.
If consumer $2$ refuses to buy one unit in period
$1$ then her utility will be unchanged. Thus, the optimal contingent pricing mechanism
gives total profits of $34$.

We now make some observations about the example described. First, observe that the equilibrium under the contingent pricing mechanism involves storage whereas the preannounced pricing solution does not. This observation is in fact general in our model (and it also holds under the model of \citet{dudine2006storable}). We prove in Section~\ref{sec:price_commitment} that regardless of the number of time periods,
the valuation functions, or whether the game is single-buyer or multi-buyer,
there is an optimal preannounced pricing strategy that induces no
storage. Thus, in terms of storage costs, the preannounced pricing mechanism is
more efficient that the contingent pricing mechanism.
Second, the prices under the preannounced pricing policy are higher than under the contingent pricing mechanism in each period. Observe as well that the retailer profits are lower under preannounced pricing ($32$) than contingent pricing ($34$). Finally, the consumer surplus (the total value obtained by consumers minus money transfers to the retailer and storage costs)
is lower under the preannounced pricing policy ($0$) than with the contingent pricing mechanism ($11$). The last three observations are surprising
as \citet{dudine2006storable} prove that none of those three statements can be true if the retailer is selling a divisible storable good. The following section provides insights about why it makes a difference whether the goods are divisible or indivisible.

\subsection{The Impact of Item Divisibility} \label{section_cont_example}
In order to provide an insight of what is the impact of good divisibility, we consider the model of \citet{dudine2006storable} and analyze an example that corresponds to the one described in the previous section. In \citet{dudine2006storable} model, there is a unit measure of consumers, and the value of consuming an amount $x \in [0,1]$ at period $t$ is $u_t(x)$, where $u_t(x)$ is an increasing and continuous function in $[0,1]$. For this example, the valuation functions $u_t(x)$ are obtained by setting $u_t(x)= \int_0^x U_t(x) dx$ where the functions $U_t(x)$ are defined as follows:

\begin{displaymath}
U_{t=1}(x) = \left\{ \begin{array}{ll}  \displaystyle 17  & \quad \textrm{if $0 \leq x \leq 0.5$} \\[3ex]
10 & \quad \textrm{if $0.5 < x \leq 1$}\\
\end{array} \right.
\end{displaymath}

\begin{displaymath}
U_{t=2}(x) = \left\{ \begin{array}{ll}  \displaystyle 15  & \quad \textrm{if $0 \leq x \leq 0.5$} \\[3ex]
4 & \quad \textrm{if $0.5 < x \leq 1$}\\
\end{array} \right.
\end{displaymath}


Under the {\em preannounced pricing policy} it easy to verify that the optimal prices and quantities bought are the same as in the divisible good settings. Specifically, at period 1, $p_1 = 17$, and a measure of 0.5 consumers buy and consume; and at period 2, $p_2=15$ and a 0.5 measure of consumers from period 2 buy and consume. Again here, no storage is used, and the retailer revenue is $\frac{17}{2} + \frac{15}{2} = 16$ \footnote{The reason the profit is not 32 as in the case where goods are not divisible is simply because of the scaling: the demand in this example is scaled to a unit per period, which is half of the example from section \ref{sec:bad-example}}.

Next consider {\em contingent pricing}. We now show that the price contingent equilibrium we found
under the indivisible model cannot be maintained in the \citet{dudine2006storable} model. The potential
equilibrium that corresponds to the equilibrium for the divisible good model is that at $t=1$, $p_1$=10 and
all the complete unit measure of consumers from $t=1$ buy, as well as 0.5 measure of consumers from period 2
(which they store for later consumption at period 2). Then, at $t=2$, we have $p_2=4$ and the remaining 0.5 measure of
consumers buy. The reason why this is not an equilibrium is because the measure of period 2 consumers who
store the good from $t=1$ to $t=2$ is so large that the retailer will then set a price $p_2=4$.
To understand this, let $y \leq 0.5 $
denote the measure of consumers who buy at $t=1$ for price 10 and store the good for $t=2$. Then,
the price at $t=2$  would still be 4 if and only if the retailer obtains more profit by charging 4 dollars to
the remaining measure consumers instead of charging 15 dollars to the remaining measure of
consumers whose valuation is 15. Therefore, we have that
$4\cdot(1-y) \geq 15 \cdot(0.5-y)$.
Solving this gives $y \geq \frac{7}{22}$.

Thus, if the retailer sets price $p_1=10$ at $t=1$, all consumers from period 1 buy and a measure of $\frac{7}{22}$ consumers of $t=2$ buy and store. Then, at period $t=2$, the retailer sets price $p_2=4$ and all remaining consumers buy. The retailer revenue under this specific contingent pricing policy is then $10\cdot(1+\frac{7}{22}) + 4\cdot \frac{15}{22} = \frac{175}{11} \simeq 15.91 < 16$. It can be shown that the optimal contingent pricing strategy is to set $p_1=17$ and $p_2=15$ so the retailer revenue under contingent pricing in this example is the same as the preannounced pricing policy (see \citet{dudine2006storable} for the general characterization of both pricing policies under this model). Table \ref{ExThreat} summarizes the differences between the two models.

\begin{table}[ht]\label{summary_difference}
\centering
\caption{Differences between Divisible and Indivisible Goods: An Example}
\begin{tabular}{|c||c|c|c|}
\hline
\textbf{Mechanism-Model} & \textbf{Retailer Profit}  &\textbf{Consumer surplus} & \textbf{Prices} \\
\hline
 \textbf{PA/Indivisible} & 32 & 0 & $p_1=17$ $p_2=15$ \\
\hline
\textbf{CP/Indivisible} & 34 & 11 & $p_1=10$ $p_2=4$ \\
\hline
\textbf{PA/Divisible}  & 16 & 0 & $p_1=17$ $p_2=15$  \\
\hline
\textbf{CP/Divisible}  & 16 & 0 & $p_1=17$ $p_2=15$  \\
\hline
\end{tabular}
\label{ExThreat}
\end{table}

In essence; using a contingent pricing policy; the retailer can sometimes extract more revenue when the goods are indivisible by performing some price discrimination that the consumer cannot overcome (under the indivisible model one unit was bought at 10, whereas the other unit was bought at 4). If the retailer tries to perform a similar price discrimination under the divisible goods setting, the consumer can finely optimize the (continuous) quantities bought, which makes the price discrimination less effective.

\section{Optimal Preannounced Pricing Mechanisms} \label{sec:price_commitment}
In this section we derive an efficient procedure to calculate an optimal
sequence of prices $\{p_1,\hdots,p_T\}$ (price at period $t$ is $p_t$) that the
retailer could announce to maximize its revenue under preannounced pricing. 
Furthermore, the optimal price sequence is the same regardless of whether we are in the
single-buyer or multi-buyer case.


We begin with the following simple lemma that holds for the multi-buyer and the single buyer settings.
\begin{lemma}\label{obs_inequality}
Let $\{p_1,\hdots,p_T\}$  be any sequence of preannounced prices. If sales at period $t$ are positive
then $p_t$ is finite and $p_t \leq p_{t'} + (t-t') \cdot c$ for all $t'<t$.
\end{lemma}
\begin{proof}
We prove this by contradiction. Suppose that at period $t$ there are positive sales (i.e. $q_t>0$) but it holds that $p_t > p_{t'} + (t-t') \cdot c$ for some $t'<t$. We consider the single-buyer case. Following (1), the consumer utility is $U = \sum_{\tau=1}^{T} \left( U(x_{\tau},\tau) -p_{\tau}\cdot q_{\tau} - c \cdot S_{\tau}\right)$ where $x_{\tau}$ denotes the items consumed in period $\tau$ and $S_{\tau} = S_{\tau-1} + q_{\tau} - x_{\tau}$ represents the number of items stored from period $\tau$ to period $\tau+1$. Suppose now that the consumer modifies her strategy composed of consumption decisions $x'_1,\hdots,x' _T$; purchasing decisions $q'_1,\hdots,q'_T$, which in turn leads to an update of the storage values $S'_1,\hdots,S'_T$. First, we leave the consumption decisions the same they were before: $x'_{\tau} = x_{\tau}$ for all $\tau=1,\hdots,T$. Second, we modify the purchasing decisions by setting $q' _{t' } = q_{t'} + q_{t}$; $q'_t = 0$ and maintaining the same purchase quantities on the remaining periods: $q' _{\tau} = q_{\tau}$ for all $\tau=1,\hdots,T$ with $\tau \neq t'$ and $\tau \neq t$. As a result of this change, the storage quantities also change as follows: $S'_{\tau} = S_{\tau}$ for all $\tau=1,\hdots,t' -1$, $S'_{\tau} = S_{\tau} + q_{t}$ for all $\tau=t',\hdots,t -1$, and $S'_{\tau} = S_{\tau}$ for all $\tau=t,\hdots,T$. In words, this change consists of bringing forward the period $t$ purchases to period $t'$. The utility obtained using the new strategy is
\begin{eqnarray*}
 U' &=& \sum_{\tau=1}^{T} \left( U(x_{\tau},\tau) -p_{\tau}\cdot q'_{\tau} - c \cdot S'_{\tau}\right)\\
 &=& \sum_{\tau=1}^{T} \left( U(x_{\tau},\tau)\right)  - \sum_{\tau=1}^{t'-1}  \left(p_{\tau}\cdot q_{\tau} + c \cdot S_{\tau}\right) - \sum_{\tau=t'}^{t}  \left(p_{\tau}\cdot q'_{\tau} + c \cdot S'_{\tau}\right) \\
 & & - \sum_{\tau=t+1}^{T}  \left(p_{\tau}\cdot q_{\tau} + c \cdot S_{\tau}\right)
\end{eqnarray*}
Where the second equality follows because $q'_{\tau}=q_{\tau}$ and $S'_{\tau}=S_{\tau}$ for all $\tau=1,\hdots,t'-1, t+1,\hdots,T$.

We will now show that $U'>U$. Indeed:

\begin{eqnarray*}
U' - U &=& \sum_{\tau=t'}^{t}  \left(p_{\tau}\cdot q_{\tau} + c \cdot S_{\tau}\right) - \sum_{\tau=t'}^{t}  \left(p_{\tau}\cdot q'_{\tau} + c \cdot S'_{\tau}\right) \\
&=& p_{t'}\cdot q_{t'} + c \cdot S_{t'} + \sum_{\tau=t'+1}^{t-1}  \left(p_{\tau}\cdot q_{\tau} + c \cdot S_{\tau}\right) + p_{t}\cdot q_{t} + c \cdot S_{t}  \\
& & - \left(p_{t'}\cdot q'_{t'} + c \cdot S'_{t'} \right) - \sum_{\tau=t'+1}^{t-1}  \left(p_{\tau}\cdot q'_{\tau} + c \cdot S'_{\tau}\right) - \left(p_{t}\cdot q'_{t} + c \cdot S'_{t} \right) \\
&=& p_{t'}\cdot q_{t'} + c \cdot S_{t'} + \sum_{\tau=t'+1}^{t-1}  \left(p_{\tau}\cdot q_{\tau} + c \cdot S_{\tau}\right) + p_{t}\cdot q_{t} + c \cdot S_{t}  \\
& & - \left(p_{t'}\cdot (q_{t'}+q_{t}) + c \cdot (S_{t'} + q_{t})  \right) - \sum_{\tau=t'+1}^{t-1}  \left(p_{\tau}\cdot q_{\tau} + c \cdot (S_{\tau}+ q_{t}) \right) - \left(p_{t}\cdot 0 + c \cdot S_{t} \right) \\
&=& -p_{t'}\cdot q_{t} - c \cdot q_t - \sum_{\tau=t'+1}^{t-1}  \left(c \cdot q_{t} \right) + p_{t}\cdot q_{t} \\
&=& q_t \left( p_t - p_{t'} - (t- t') \cdot c\right) \\
&>& 0.
\end{eqnarray*}

Above, the third equality is obtained by replacing the new purchased quantities $q'$ and the new storage quantities $S'$ as a function of the original strategy. Finally, the strict inequality at the end is due to our assumption that $p_t > p_{t'} + (t-t') \cdot c$.
This is a contradiction, as the consumer would have obtained a strictly higher utility $U' $ by switching to the new proposed strategy.

We note that the argument used to prove the lemma also applied to the multiple buyer case. Indeed, in that setting each of the consumers who bought at period $t$ can be better off by moving forward the purchase to period $t'$ for exactly the same reasons.
\end{proof}


The following lemma shows that under preannounced pricing there is no wasteful storage. This ``storage-proof" property
of the preannounced pricing mechanism extends the result of \citet{dudine2006storable} to the case of
indivisible demand.

\begin{lemma}\label{lemma_buy_on_time}
There always exists a revenue maximising preannounced pricing strategy in which there is no storage. This
pricing strategy is the same for the single-buyer and the multiple-buyer settings.
\end{lemma}

\begin{proof}
Take any optimal preannounced pricing strategy $\{p_1,\hdots,p_T\}$ in which the consumer (or consumers)
use the storage option. We show that there is another optimal  preannounced pricing strategy
that requires less storage. The theorem will follow.
Let $t$ be the first time period where $q_t>x_t$. Thus a consumer stores at least one unit of the good
between period $t$ and period $t+1$.
We may assume in the case of a tie that the consumer(s)
selects an option with the least storage.
By the rationality of each consumer, we then have
that $p_{t+1}> p_t + c$. By Lemma~\ref{obs_inequality}, we then have that $q_{t+1}=0$.
Now suppose the retailer modifies its preannounced policy by setting $p_{t+1}  = p_{t} + c$.
Observe first that consumption in periods $1$ to $t$ will not change. Second, in the single-buyer case
the consumer now prefers
to purchase $q_t-x_t$ units in period $t+1$ rather than buying them earlier and storing them. In the
multi-buyer case there are $q_t-x_t$ consumers who previously bought a unit to store at period $t$ and
now prefer to purchase it in period $t+1$.
Thirdly, the purchases made from period $t+2$ are unaffected by this price change; if it is now preferable to buy
additional units in period $t+1$ then each consumer would have originally benefited by making such purchases
in period $t$. This modified pricing strategy cannot decrease the retailer profits (if $c>0$ then profits strictly increase).
\end{proof}



Now recall that in the multiple-buyer setting, $v_{i,t}$ is the value of the $i^{th}$ consumer (ranked in decreasing
order of their valuation in time period $t$). In the single-buyer setting, in order to unify notation, we let $v_{i,t} = V(i,t)$
(that is the marginal value of consuming $x$ units of the good at period $t$) and let the value of $H \in \mathbb{N}$
such that $V(H,t)=0$ for all $t \in T$ be equal to $N+1$.

Lemma \ref{lemma_buy_on_time} allows us to write the retailer's problem as follows. The formulation is valid
for the single-buyer model as well as for the multiple-buyer setting.
\begin{eqnarray*}
\max\ \ \sum_{t=1}^T p_t \cdot q_t &&\\
s.t.\hspace{2cm}
p_t &\leq& v_{q_t,t}  \hspace{.25cm} \forall t \\
p_{t} &\le& p_{t'} + (t-t')\cdot c \hspace{.25cm} \forall t'<t   \hspace{.5cm} \mathrm{{\bf if}}\  q_t>0 \\
q_t &\in& \mathbb{Z}^+ \hspace{.5cm} \forall t
\end{eqnarray*}
Here $q_t$ is a non-negative integer representing the number of items sold at time $t$.
For the consumer to desire this amount, the price at time $t$ must be at most
the corresponding marginal value -- this gives the first constraint.
The second constraint
states that if there are sales in period $t$ then the price $p_t$ cannot be more than the total cost of purchasing
a good in an earlier period and storing it until time $t$.
Recall that we set $v_{0,t}=L$. Thus, if $q_t=0$ we may set $p_t=L$; this cannot decrease
the objective function as it can only weaken the restrictions imposed by the second set of constraints.

The following theorem is crucial for the construction of the polynomial-time algorithm to compute an optimal preannounced pricing strategy.

\begin{theorem}\label{theorem_1}
There exists optimal preannounced pricing sequence $p_1,\hdots, p_T$ such that
for each $t$, we have $p_t = v_{j,s} + (t - s)\cdot c$ for some $s \leq t$ and some $0\le j\le N$.
\end{theorem}
\begin{proof}
Take an optimal sequence of preannounced pricing sequence $(p_1,\hdots, p_T)$ that do not require storage.
If these do not satisfy the requirements of the theorem then take
smallest time $t$ such that $p_t \neq v_{j,s} + (t - s)\cdot c$ for any $s \leq t$ and for any $j \leq N$.
Now, if $q_t=0$ then set $\hat{p}_t = v_{0,t}=L$. It is easy to verify that
the new prices $\{p_1,\dots, \hat{p}_t,\dots, p_T\}$ induce the same behaviour by the consumer or consumers
and, thus, produce the same revenue.
Moreover, the period $t$ now satisfies the requirements
of the theorem for $s=t$ and $j=0$.
So suppose $q_t\ge 1$. We now set
$$\hat{p}_t = \min \{v_{j,s} + (t - s)\cdot c : s \leq t;  0\le j\le N; v_{j,s} + (t - s)\cdot c > p_t\}$$
We claim  $(p_1,\dots, \hat{p}_t,\dots, p_T)$ is a more profitable sequence of preannounced prices. To do this we again desire that the consumer (or consumers) behaviour does not change.
If so, the price increase at time $t$ will result in additional revenue.
First, the price increase at time $t$ will not affect consumption by the consumer(s)
in any period $t'<t$. Second the price increase at time $t$ will not alter consumption in periods
$t''>t$; making additional purchases at time $t$ and
storing them for future consumption is even less desirable than before.
It remains to consider sales in period $t$.
 Recall there were originally $q_t$ sales in period $t$.
 Because there was no storage, it must, by consumer rationality, be the case
 that $v_{q_{t+1},t} < p_t \le v_{q_{t},t}$.

Thus, $p_t < v_{q_{t},t}$, otherwise $p_t$ does satisfy the requirements of the theorem.
It follows that $v_{q_{t},t}$ was a feasible option when selecting $\hat{p}_t$.
Hence, $\hat{p}_t\le v_{q_{t},t}$.
Thus the price of the item at time $t$ is still at most the marginal value the $q_t^{th}$ unit (in the multi-buyer
case, $\hat{p}_t$ is at most the value at period $t$ of the consumer with the $q_t^{th}$ highest valuation in that period).
But does the price increase from $p_t$ to $\hat{p}_t$ now induce the consumer or consumers to now purchase early
and store for consumption at time $t$? This is not the case. Since $q_t>0$, we know by Lemma~\ref{obs_inequality}
that $p_t \leq p_{s}+ (t - s)\cdot c$ for all $s<t$. Since by assumption $p_t \neq v_{j,s} + (t - s)\cdot c$ for any $s \leq t$ and
for any $j \leq N$, we have that
$p_t <  p_s + (t -s)\cdot c$.
Now, by the minimality of $t$ we have that $p_{s} = v_{j',r}+ (s - r)\cdot c$ for some $r\le s$ and $0\le j'\le N$. Hence
\begin{equation*}
p_t  \  < \  p_s + (t -s)\cdot c
 \ =\  v_{j',r}+ (s - r)\cdot c + (t-s)\cdot c
 \ =\   v_{j',r}+ (t - r)\cdot c
\end{equation*}
It then follows, by the definition of $\hat{p}_t$, that
$\hat{p}_t \leq p_{s}+ (t - s)\cdot c$.
Consequently, it is not beneficial for the consumer or consumers to use storage and
$\{p_1,\dots, \hat{p}_t,\dots, p_T\}$ is a more profitable solution for the retailer.
Iterating the argument produces a sequence of preannounced prices that
satisfy the requirements of the theorem for every period $t$.
\end{proof}

The structure provided by Theorem~\ref{theorem_1} allows us to construct an
optimal preannounced pricing strategy in polynomial time via a dynamic
program. The key is that we search for prices of the form $p_t = v_{j,s} + (t - s)\cdot c$,
for some $s \leq t$ and some $0 \le j \le N$. Intuitively, this means we are
matching a price $p_t$ with a value $v_{j,s}$. Of course, this correspondence
has to satisfy numerous constraints in order to be rational and profit-maximizing.

To explain how to find the optimal correspondence, we begin with
some definitions.
For each demand value $v_{i,t}$ we associate a 2D {\em contour} $\gamma_{i,t}$ to the demand of the $i$th item in period $t$.
(Recall that demands are sorted by decreasing value in each period, that is, $v_{i,t} \geq v_{i+1,t}$.)
We then define a total ordering $\preceq$ over all the contours. Specifically, we say $\gamma_{i,t} \preceq  \gamma_{j,r}$
if $v_{i,t} + c \cdot(r-t) \leq v_{j,r}$ and $t \leq r$.
To motivate this, assume the consumer is debating whether to purchase and consume $j$ items in period $r$
or to purchase in period $t<r$ and then store the goods until period $r$.
The best choice is determined by the ordering of the contours $\gamma_{i,t}$ and $\gamma_{j,r}$.

We denote by $\mathcal{C}_t=\{\gamma_{i,t} : i\in [N]\}$  the set of contours associated with demands in period $t$.
Further, we set $\mathcal{C}_0$ to contain a dummy contour, $\mathcal{C}_0= \{ \gamma^*\}$ where $\gamma  \preceq \gamma^*$
for all $ \gamma \in \bigcup_{t=1}^{T}\mathcal{C}_t$. We will view the selection of the dummy contour in time $t$
as the assignment $p_t=L$ for some huge~$L$.

Now suppose we select a contour $\gamma_{j,s}$ at time $t$. As described,
this produces the price $p_t = v_{j,s} + (t - s)\cdot c$. In addition, this choice
constrains the price choices for all periods $r>t$. Specifically, such a price $p_r$
must correspond to a contour no higher than $\gamma_{j,s}$. Hence, the choice
of contours (excluding dummy contours) is non-increasing over time.
To formalize this for the dynamic program we use the following notation.
Define $\mathcal{F}_t(\gamma_{i,y})= \{ \gamma \in \mathcal{C}_t\cup \{ \gamma_{i,y}\}: \gamma \preceq  \gamma_{i,y} \}$
where $\mathcal{F}_t(\gamma_{i,y})$ is the \emph{feasible set} of contours for period $t$, given that the lowest contour chosen so
far is $\gamma_{i,y}$.  Given a contour  $\gamma_{i,s}$ and a time period $t \geq s$, we
define $p_t(\gamma_{i,s}) = v_{i,s} + c\cdot(t-s)$.
We then set $q_t(\gamma_{i,s}) = \#\{j\ge1: v_{j,t} \ge p_t(\gamma_{i,s})\}$ to be the number of items in period $t$ with value at least
$p_t(\gamma_{i,s})$. Note that if there is no storage then $q_t(\gamma_{i,s})$  is exactly the number of sales in period $t$
when prices are $p_t(\gamma_{i,s})$.
We are now ready to present a dynamic program that calculates the optimal sequence of preannounced prices in polynomial time.


\begin{algorithm}\label{algorithm_price_commitment}
\label{no-trust-general}
\caption{Optimal preannounced pricing mechanism}
\begin{algorithmic}
\STATE $\mathcal{R}(T+1, \gamma) \gets 0$ for all $\gamma \in \bigcup_{\tau=1}^T\mathcal{C}_{\tau}$.
\FOR{$t = T \to 1$}
    \FOR{each $\gamma_{i,s} \in \bigcup_{\tau=0}^{t-1}\mathcal{C}_\tau$}
    \STATE $\mathcal{R}(t, \gamma_{i,s}) \gets \max \{ q_t(\gamma_{j,r}) \cdot p_t(\gamma_{i,r})
    + \mathcal{R}(t+1, \gamma_{j,r}) :  \gamma_{j,r} \in \mathcal{F}_t(\gamma_{i,s})\}$

    \STATE $\mathcal{S}(t, \gamma_{i,s}) \gets \arg \max_{\gamma_{j,r}} \{ q_t(\gamma_{j,r}) \cdot p_t(\gamma_{i,r})
    + \mathcal{R}(t+1, \gamma_{j,r}) :  \gamma_{j,r} \in \mathcal{F}_t(\gamma_{i,s})\}$

        \ENDFOR
\ENDFOR
\RETURN $\{p_1(\mathcal{S}(1,\gamma^*)),p_2(\mathcal{S}(2,\mathcal{S}(1,\gamma^*))),\hdots\}$

\end{algorithmic}
\end{algorithm}

\begin{theorem} The dynamic program computes an optimal preannounce pricing strategy.
\end{theorem}
\begin{proof}
We want to show that $\mathcal{R}(t, \gamma_{i,s})$ computes the optimal revenue obtainable
in period $t$ to $T$, given that the maximum price we may use in period $t$ is
induced by the contour $\gamma_{i,s}$. That is, the maximum price $p_t$ is
at most $v_{i,s}+(t-s)\cdot c$. We prove this by backwards induction from the
last period. For $t=T$, we have that $\mathcal{R}(t, \gamma_{i,s})$
maximises the quantity times price given the upper bound on price implied by $\gamma_{i,s}$, as desired.
Now consider $\mathcal{R}(t, \gamma_{i,s})$, for $t<T$.
Again, by Theorem \ref{theorem_1} there exists an optimal pricing strategy in which the price at
period $t$ is equal to $p_t(\gamma_{i,r})$ for some $\gamma_{i,r} \in \bigcup_{\tau=1}^{t}\mathcal{C}_\tau$.
To maximise profits we wish to select a contour in $\mathcal{C}_t$; however, the corresponding price
must not be above $\gamma_{i,s}$. Thus, the optimal choice, $\mathcal{S}(t, \gamma_{i,s})$, must either
be $\gamma_{i,s}$ itself or a lower contour in $\mathcal{C}_t$. The objective is to maximise
revenue in period $t$ plus revenues from periods $t+1$ to $T$. The correctness of the procedure
follows as the latter has been calculated inductively.
 \end{proof}

The following lemma shows that the algorithm presented above is very efficient in terms of time requirement.

\begin{lemma}
An optimal sequence $(p_1\hdots,p_T)$ of preannounced pricing can be computed in time $O(T \cdot D^2)$,
where $T$ is the number of periods and $D$ is the total number of items in demand ($D \leq N\cdot T$).
\end{lemma}
\begin{proof}
The outer loop covers $T$ periods. The inner loop takes examines at most $D$ counters (one for each demand item).
The two variable assignments require $O(D)$ operations. Together this takes time $O(T \cdot D^2)$.
\end{proof}
We end this section by noting that for the model with divisible goods, \cite{dudine2006storable} provided another dynamic program to compute efficiently the sequence of preannounced prices. One of appealing features of the algorithm presented here is the flexibility to work for arbitrary consumer valuations.\footnote{We remark that the model in \citet{dudine2006storable} contains additional assumptions that are not
required here, such as the retailer revenue being a concave function of the price in each time period.}

\section{A Relationship between Preannounced Pricing Profits and Contingent Pricing Profits}
Recall, in the divisible storable good \citep{dudine2006storable},
retailer profits under a preannounced pricing mechanism are always higher than or equal to
those obtained under a contingent pricing mechanism. We saw an example in Section~\ref{sec:bad-example} showing
this need not be the case for an indivisible good. In this section, we quantify exactly how much more profits
a contingent pricing mechanism can generate than a preannounced pricing mechanism for an indivisible good.
Specifically, let $\Pi^{CP}(\mathcal{G})$ and $\Pi^{PA}(\mathcal{G})$ denote the retailer revenue
in the game $\mathcal{G}$ under contingent pricing and preannounced pricing respectively.
Our main result is that, for the multi-buyer case with $N(\mathcal{G})$ consumers,
there exists an example for which $\Pi^{CP}(\mathcal{G}) = \Omega(\log(N(\mathcal{G})) \cdot \Pi^{PA}(\mathcal{G}))$.
Furthermore this result is tight, for any game  $\mathcal{G}$, we have
$\Pi^{CP}(\mathcal{G})= O(\log(N(\mathcal{G})) \cdot \Pi^{PA}(\mathcal{G}))$.
Before presenting these results, we give a simple example to show that
in the single-buyer case there are examples where the
retailer can obtain arbitrarily higher profits using a contingent pricing mechanism rather than
a preannounced pricing mechanism.

\subsection{The Single-Buyer Case.}
Here is a single-buyer example where profits for the contingent pricing mechanism are arbitrarily higher than for the preannounced pricing mechanism

\begin{theorem} \label{unbounded_NC_profits_single_buyer}
For any $k \in \mathbb{N}$, there exists a single-buyer storable good game $\mathcal{G}$
such that $\Pi^{CP}(\mathcal{G}) > k \cdot \Pi^{PA}(\mathcal{G})$.
\end{theorem}
\begin{proof}
We create a game $\mathcal{G}$ with $T=2$ periods. The buyer receives zero value from consumption
in the first period. In the second period, the buyer has a marginal value of $1+\epsilon$
for the first unit of consumption (for some small $\epsilon >0$), and a marginal value of $\frac{1}{i}$ for the
$i$th unit of consumption, for $2\le i\le N$. This is summarized in
Table \ref{single_buyer_example_no_commitment_table}. The storage
costs are assumed to be negligible, that is, $c=0$.

\begin{table}
\begin{center}
\caption{Single-Buyer Example with Profits arbitrarily Higher under Contingent Pricing
than under Preannounced Pricing}\label{single_buyer_example_no_commitment_table}
\begin{tabular}{|c || c | c |c | c | c | c | c |}
\hline
Item & 1& 2& 3& $\cdots$ &i& $\cdots$& N\\
\hline
Marginal Value in Period $t=1$ &0&0& 0& $\cdots$ &0& $\cdots$ & 0 \\
Marginal Value in Period $t=2$ &1+$\epsilon$ & $\frac12$ & $\frac13$ &$\cdots$&$\frac{1}{i}$&$\cdots$& $\frac{1}{N}$ \\
\hline
\end{tabular}
\end{center}
\end{table}


Using a preannounced pricing strategy the maximum profit obtainable by the retailer
is clearly $\Pi^{PA}(\mathcal{G}) = 1 + \epsilon$. On the other hand, the retailer can do far better
by using a contingent pricing mechanism.
To see this, first consider what happens if the consumer refuses to buys any items
in the first period. The price in the second period will be then be the static monopoly price of $1+\epsilon$;
the consumer surplus would then be zero.
So suppose the retailer announces a price of
$p_1 = H_N + \frac{1}{N} -1+\epsilon$ in the first period (where  $H_N = \sum_{i=1}^{N}\frac{1}{i}$). As discussed, the consumer generates no surplus if she refuses to buy
in period 1. What happens if the consumer buys one unit at price $p_1$ and stores it for consumption in period $2$?
Then, in the second period, if the consumer buys $k<N$ additional items, her marginal value for that final unit of consumption is $\frac{1}{k+1}$. Therefore, given that the consumer bought exactly one unit in period $1$,
the optimal price for the retailer is  $p_2 = \frac{1}{N}$ in the second period. Consumer surplus over the two
periods is then
\begin{equation*}
\left(1+\epsilon +\sum_{i=2}^{N}\frac{1}{i}\right) - \left(1\cdot p_1 + (N-1)\cdot p_2\right)
\ =\
\left(H_N+\epsilon\right) - \left(H_N + \frac{1}{N} -1+\epsilon + \frac{N-1}{N}\right)
\ =\ 0
\end{equation*}

Hence the consumer is no worse off by purchasing one item in the first period at $p_1$ than if she waits and buys
only in the second period. The consumer has negative utility if she purchases more than one unit in period one.
Thus, we have an equilibrium where the consumer purchases exactly one unit in period one with
retailer's profits of
\begin{equation*}
\Pi^{CP}(\mathcal{G})
\ =\ 1\cdot p_1 + (N-1) \cdot p_2
\ =\ H_N + \frac{1}{N} -1+\epsilon + \frac{N-1}{N}
\ =\ H_N +\epsilon
\end{equation*}

It is well-known that $H_N > \log N +\frac12$. Consequently, we
have $\Pi^{CP} \ge \log N\cdot \Pi^{PA}$, as $\epsilon \rightarrow 0$.
\end{proof}

%

\subsection{The Multi-Buyer Case.}
Is it also the case that retailer profits can be arbitrarily higher under contingent pricing mechanisms
than under preannounced pricing mechanisms when there are a large number of consumers?
We remark that the example used to prove Theorem \ref{unbounded_NC_profits_single_buyer}
does not extend to the multi-buyer case.
To see this, suppose we interpret Table~\ref{single_buyer_example_no_commitment_table}
as a game with $N({\mathcal{G}})=N$ buyers each with a unit demand. The optimal prices using the contingent pricing mechanism would then be $p_1= 1 +\epsilon$ (one consumer buys) and $p_2 = 1/N$ (and the
 remaining $N-1$ consumers buy). But this only yields a total profit that is less than $2\cdot \Pi^{PA}$.
 Intuitively, the reason why higher profits cannot be extracted is because the $N$ buyers do
 not cooperate. In contrast, in the single-buyer setting, the retailer is able to exploit the fact that
 the buyer is willing to buy at a ``loss" in period~1 in order to obtain a surplus in period~2.

Nonetheless, we now show, via a far more complex example, that the retailer in the multi-buyer setting
can generate a multiplicative factor $\log(N({\mathcal{G}}))$ more profits by using a contingent pricing mechanism.
This is interesting, and perhaps surprising, in light
of the results of \citet{berbeglia2014} who examined a model of a durable good monopoly problem with
$N$ atomic consumers and finite time horizon originally proposed by \citet{bagnoli1989durable}.
In that model, \citet{berbeglia2014} proved that retailer profits under a contingent pricing mechanism are always within
an additive constant (and within at most a multiplicative factor of two)
of profits under a preannounced pricing strategy.
The model proposed by \citet{bagnoli1989durable} corresponds to a special case of the storable
good problem studied here. Specifically, to capture their model we simply
assume each consumer has a positive valuation only in the final period, that is,
$v_{i,t} =0$ for all $1\le i\le N$ and for all $1\le t\le T-1$.\footnote{Essentially, we can view the
all the consumers in the \citet{bagnoli1989durable} model as only wishing to consume
the durable good in the final period. In those circumstances, it does not matter whether the good is durable or simply storable.}
This additive constant bound does not extend to the more general storable goods problem studied in this paper. Again, as in the single-buyer case, the ratio between the profits
under a contingent pricing mechanism over the profits under a preannounced pricing strategy can be unbounded.

\begin{theorem} \label{log_factor_thm}
There exists an infinite
sequence of games $\mathcal{G}_1,\mathcal{G}_2,\hdots$, such that the number of
consumers is strictly increasing and that
$ \frac{\Pi^{CP}(\mathcal{G}_i)}{\Pi^{PA}(\mathcal{G}_i)} = \Omega(\log(N(\mathcal{G}_i)) = \Omega(\log T)$.
\end{theorem}
\begin{proof}
Take any integer $n>2$. We create a game with $T= \sum_{i=1}^n2^{n-i} = 2^n-1$ periods.
In each period $t \in [T]$, there is exactly one consumer with positive valuation for the item; that is, $v_{1, t}>0$ and $v_{i,t}=0$,
for all $1\le t \le T$ and all $2\le i \le n$.
Furthermore, each consumer has a positive valuation for the good in only one period.
Specifically, we partition  the time periods into $n$ blocks, $\{\mathcal{B}_1, \mathcal{B}_2,\dots, \mathcal{B}_n\}$.
Here block $\mathcal{B}_k$ consists of $2^{n-k}$ consecutive time periods.
Note that $\sum_{k=1}^n|\mathcal{B}_k| = 2^n-1$ and so there are indeed $T=2^n-1$ time periods in
total.
It remains to assign, at each period $t$, a valuation to the only consumer who wishes to consume in that period.
If consumer $i$ is the unique buyer with positive valuation in period $t$ then we say that $i$ \emph{is} the
consumer of period $t$, or that $i$ \emph{belongs} to period $t$. (Given a set $S$ of time periods, we say
that a consumer \emph{belongs} to $S$ if the consumer belongs to any period in $S$.) The consumer belonging to
any period in block $\mathcal{B}_k$, for each $1\le k\le n$, is given a valuation of $v_k=2^{k-1}$.
Again, we assume that storage costs are negligible, that is $c=0$.
Table \ref{tb_game_construction} summarizes this construction.

\begin{table}\label{tb_game_construction}
\caption{Block Sizes and Valuations in the Multi-Buyer Game}
\begin{center}
\begin{tabular}{|c || c | c |c | c | c | c | c |}
\hline
Block Number & 1& 2& 3& $\cdots$ & s & $\cdots$ & $n$\\
\hline
Time Length & $2^{n-1}$ & $2^{n-2}$ & $2^{n-3}$& $\cdots$ & $2^{n-s}$ & $\cdots$ & $2^{n-n}=1$ \\ \hline
Consumer Value & 1& $2^1$& $2^2$& $\cdots$ & $2^{s-1}$ & $\cdots$ & $2^{n-1}$ \\ \hline
\end{tabular}
\end{center}
\end{table}

Let's first calculate the retailer profit, $\Pi^{PA}(\mathcal{G})$, under preannounced pricing.
To do this, we claim that there is an optimal preannounced pricing strategy in
which prices do not change over time. Let $\{p_1, p_2,\dots, p_T\}$ denote an optimal preannounced pricing strategy. Consider the largest $t$ such that $p_t < p_{t+1}$. Because there are no storage costs
no consumer in any period $s$ with $s>t$ would buy at price $p_{t+1}$. Thus, setting a price
$p' _{t+1}= p_t$ at period $t+1$ will not cause the retailer's profits to decrease.
On the other hand, take the largest $t$ such that $p_t > p_{t+1}$. If no consumer buys in period $t$ then
decreasing the price in period $t$ down to $p' _{t}=p_{t+1}$ can only increase profits.
If at least one consumer does buy in period $t$ at price $p_t$ then that must be
the consumer that belongs to period $t$; any other consumer will do better by waiting to purchase
in period $t+1$. This means that the consumer from period $t$ has a valuation
at least $p_t$. By the construction of the game, this in turn implies that every consumer belonging
to a period $s$, with $s>t$, also has a
valuation at least $p_t$. Therefore, even with a fixed price of $p_t$ between periods $t+1$ to $T$ all consumers in periods $t+1$ to $T$ would still buy the item. This increases monopoly profits.
It follows that there is a fixed price strategy that yields optimal profits.
Therefore, to compute $\Pi^{PA}(\mathcal{G})$ we can restrict attention to fixed price strategies.
It is then easy to verify that the optimal fixed price is $1$. Every consumer will then purchase
the good giving the retailer a profit of $1$ in each time period, and thus
of $T=\sum_{i=0}^{n-1} 2^i = 2^n - 1$ in total. (Note also that $N(\mathcal{G})=T=2^n-1$). Hence $\Pi^{PA}(\mathcal{G}) = 2^n - 1$.


Next consider contingent pricing mechanisms.
We will show that there is a subgame perfect equilibrium in which the retailer extracts
almost all the consumer surplus. This equilibrium is composed of the following pair of
strategies. The strategy of the retailer is to announce a price equal to the value of the highest
consumer yet to buy. The strategy of each consumer is to buy the item in the earliest period $t$ where
the price is at most her valuation, provided the consumer belongs to a period $s \ge t$.
Following \citet{bagnoli1989durable}, who studied a similar strategy pairing
in a different setting, we call the retailer's strategy \emph{pacman} and the consumers
strategies \emph{get-it-while-you-can}.

To analyse this we denote by $\mathcal{G}_{i,k}$ a subgame in which the remaining consumers
belong to the set of blocks $\{\mathcal{B}_{i}, \mathcal{B}_{i+1}, \dots, \mathcal{B}_k\}$. That is,
the consumers belong to a subset  $S \subseteq \bigcup_{j=i}^k \mathcal{B}_j$. Without loss of generality,
we may assume that $S$ contains at least one consumer of the block $\mathcal{B}_k$
and that the subgame begins at a period $t$ associated with block $\mathcal{B}_i$.
We wish to prove that, in any subgame $\mathcal{G}_{i,k}$, the \emph{pacman} strategy is a
best response to the consumer strategy \emph{get-it-while-you-can}. We proceed by induction on $k-i$.
If $k-i=0$ then all the remaining consumers have a fixed value. It follows immediately
that the \emph{pacman} strategy is a best response to \emph{get-it-while-you-can}.
Suppose that there are two blocks, that is, $k-i =1$. Let $|\mathcal{\tilde{B}}_i|$, respectively $|\mathcal{\tilde{B}}_k|$,
denote the number of active consumers in block $i$, respectively $k$. Note that $|\mathcal{\tilde{B}}_j| \le |\mathcal{B}_j|$
as the set $S$ in the subgame $\mathcal{G}_{i,k}$ need contain only a subset of the original collection of customers
in block $\mathcal{B}_j$. Using the pacman strategy then yields profits at least:
\begin{equation*} \label{base_case_eq}
|\mathcal{\tilde{B}}_k| \cdot v_k + (|\mathcal{\tilde{B}}_i|-1) \cdot v_i
\ =\ |\mathcal{\tilde{B}}_k| \cdot 2 \cdot v_i + (|\mathcal{\tilde{B}}_i|-1) \cdot v_i
\ =\ (2 \cdot |\mathcal{\tilde{B}}_k| + |\mathcal{\tilde{B}}_i|-1 )\cdot v_i 
\end{equation*}
Observe that in period $2$, of the subgame $\mathcal{G}_{i,k}$, we may make only $|\mathcal{\tilde{B}}_i|-1$ sales
as one of the lower value consumers may only wish to purchase in the first period $t$ of the subgame.
The only other logical strategy against get-it-while-you-can is for the retailer to announce a price of $v_i$  in
the first period $t$. This would cause each consumer to
buy in the first time period giving a profit of $(|\mathcal{\tilde{B}}_k| + |\mathcal{\tilde{B}}_i|)\cdot v_i$.
But $2 \cdot |\mathcal{\tilde{B}}_k| + |\mathcal{\tilde{B}}_i|-1\ge |\mathcal{\tilde{B}}_k| + |\mathcal{\tilde{B}}_i|$, and
so the pacman strategy is at least as profitable as this.

Before proceeding with the inductive step, denote by $\tau^P$ the number of periods required by
the pacman strategy to sell to (almost) every active consumer in $\mathcal{G}_{i,k}$. Observe that $\tau^P \leq k-i+1$
but that some consumers may not be able to buy the good. Specifically, the consumers (if any) that do not
buy all belong to the first $k-i$ time periods, that is from period $t$ up to period $t+k-i-1$. All those periods
are associated with blocks $\mathcal{B}_i$ and $\mathcal{B}_{i+1}$, since the number of periods in these
two blocks adds up to at least $1 + 2^{k-i-1}$, which is not less than $k-i$ for all $k\geq 2$.

Now if the retailer announces a first price of $v_{\alpha}=2^{\alpha-1}$, for some $\alpha<k$, then the
active consumers of blocks $\mathcal{B}_{j}$
with $j \geq \alpha$ will buy due to the \emph{get-it-while-you-can} strategy. Period $t+1$
then begins a subgame with at least one fewer block in which, by induction, pacman is the
best response. Let $\Pi^{\alpha}(\mathcal{G}_{i,k})$ denote the retailer revenue obtained by announcing
a first period price of $v_{\alpha}$ ($\alpha<k$). We can now write  necessary and sufficient conditions for the pacman
strategy to provide a higher or equal revenue than any such deviation. Let $\tau^{\alpha}$ be the
number of periods needed to sell to all the consumers (that buy) if the retailer announces a first
period price of $v_{\alpha}$ and then follows the pacman strategy. Thus,
$\tau^{\alpha} \geq \tau^{P} - (k-\alpha)$.
The corresponding additional surplus generated by selling to potentially more consumers is at most
\begin{equation*}
(k-\alpha)\cdot \min[v_{\alpha}, v_{i+1}]
\end{equation*}
To see this, recall that these $k-\alpha$ consumers belong to blocks $\mathcal{B}_i$ and $\mathcal{B}_{i+1}$.

On the other hand, by setting an initial price of $v_{\alpha}$ the retailer cannot obtain full value from
the consumers in the block $\mathcal{B}_k$. There is at least one such consumer, so this causes the
retailer to lose at least $v_k-v_{\alpha}$.
To prove that pacman is a best response, it then suffices to show that the potential gains of announcing a
price $v_{\alpha}$ do not compensate the losses. To show this, we have
 \begin{eqnarray} \label{diff_revenue}
\Pi^{pacman}(\mathcal{G}_{i,k}) - \Pi^{\alpha}(\mathcal{G}_{i,k})
&\geq& (v_k -v_{\alpha})- (k-\alpha)\cdot \min[v_{\alpha}, v_{i+1}] \nonumber \\
&=& (2^{k-\alpha}-k+\alpha-1)\cdot v_{\alpha} \nonumber \\
&\geq& 0 \vspace{10 mm}  \hspace{1.0cm} \mbox{\ [as $k-\alpha\geq 1$]\ }  \nonumber
\end{eqnarray}
It follows that we have a subgame perfect Nash equilibrium.
In this equilibrium, the retailer achieves near perfect price discrimination. Every consumer
except (possibly) the first $n-1$ consumers buy the good at exactly their valuation. The revenue obtained
is therefore
\begin{equation*}
\Pi^{CP} \ = \ \sum_{i=1}^n 2^i\cdot 2^{n-i} - (n-1)\cdot 1
\ =\  n\cdot 2^n-(n-1)
\ \ge\  n\cdot (2^n-1)
\ =\  n\cdot \Pi^{PA}
\end{equation*}
Recall that the number of consumers in this game is $N(\mathcal{G}_i)=T=2^{n} - 1$. Thus
$$\Pi^{CP} \ge \log N(\mathcal{G}_i) \cdot \Pi^{PA} $$ as desired.
\end{proof}


The following lemma shows that Theorem \ref{log_factor_thm} is tight.
\begin{lemma}\label{lem:log_factor_thm}
 For any game $\mathcal{G}$, $\frac{\Pi^{CP}(\mathcal{G})}{\Pi^{PA}(\mathcal{G})} = O(\log(N(\mathcal{G})) +\log T)$.
\end{lemma}
\begin{proof}
Let $\Pi^{F}(\mathcal{G})$ denote the optimal profit that can be obtained by the retailer in $\mathcal{G}$ by setting a
fixed price along the game. Clearly, $\Pi^{PA}(\mathcal{G}) \geq \Pi^{F}(\mathcal{G})$.
Let $v_1, v_2, \dots,v_{\ell}$ denote, in decreasing order, the set of all consumer valuations
over all time periods. Observe that ${\ell} \le N(\mathcal{G})\cdot T$.
The contingent pricing mechanism cannot extract more profit than a perfect price discrimination
mechanism. Thus $$\Pi^{CP}(\mathcal{G}) \leq \sum_{k=1}^{\ell} v_k$$
Thus, it suffices to show the following inequality:
\begin{eqnarray} \label{inequality_to_prove_upper_bound}
\Pi^{F}(\mathcal{G}) \cdot H_{\ell} \geq \sum_{k=1}^{\ell} v_k,
\end{eqnarray}
where $H_{\ell} =\sum_{i=1}^{\ell} \frac{1}{i}= O( \log( N(\mathcal{G})\cdot T)) = O( \log( N(\mathcal{G}) +\log T)$.
To do this, set $j^* = \arg \max_{j}  \{ j \cdot v_{j} : j \in [l] \}$. Thus, a retailer that announces a static price of $v_{j^*}$ would
obtain a profit of $j^* \cdot v_{j^*} = \Pi^{F}(\mathcal{G})$.
Without loss of generality, scale the consumer valuations such that  $j^* \cdot v_{j^*} =1$. We then
have, for all $i \in [\ell]$, that
$i \cdot v_i \leq 1$. Hence,  $v_i \leq  \frac{1}{i}$.
Inequality (\ref{inequality_to_prove_upper_bound}) follows.
\end{proof}


\section{Concave Storage Costs}
We now examine our model under the more general setting of concave storage costs. Here the marginal cost of storage is non-decreasing function in quantity.
Let $\mathcal{C}: \mathbb{N} \to \mathbb{R}_{\geq 0}$ be an arbitrary concave function
such that $\mathcal{C}(q)$ denotes the cost of storing $q$ units of the good for one period. (Here $\mathcal{C}(0)=0$.)
Because linear costs are a special case of concave costs, the example in Theorem~\ref{log_factor_thm}
shows the contingent pricing mechanism can produce an $\Omega(\log N)$ factor more profits than the optimal preannounced pricing mechanism
in the case of concave costs. Furthermore, it can be seen that the $O(\log N)$ upper bound of
Lemma~\ref{lem:log_factor_thm} also applies in the case of concave storage cost. We now examine
the question of whether storage occurs under preannounced pricing mechanisms. The answer differs
for the single-buyer and multi-buyer settings.

\subsection{The Single-Buyer Case.}\label{sec:concave-single}
We extend Theorem \ref{lemma_buy_on_time} from linear to concave storage costs in the single-buyer setting.
Specifically, we show that there exists an optimal price-commitment strategy under which
no storage occurs.

\begin{theorem}\label{theorem_single_buyer_concave}
In the single-buyer setting where storage costs are concave, there always exists a revenue maximising price commitment
strategy that induces no storage.
\end{theorem}
\begin{proof}
Let $\mathcal{D}_1$ denote any optimal price-commitment strategy for the retailer, i.e the sequence of $T$ prices,
and let $\mathcal{D}_2$ denote the buyer best response to $\mathcal{D}_1$, i.e. an optimal sequence of the number
 of items bought in each time period. Let the storage levels under $\mathcal{D}_2$ be $S_1, S_2,\hdots, S_{T-1}$,
where $S_t$ denotes the number of units stored from period $t$ to period $t+1$ for all $t \in [T-1]$.
Let $Q(\mathcal{D}_2) = \sum_{t=1}^{T-1} S_t$ be the total amount of storage.  Without loss of generality, we assume
that given two purchasing schedules that yield the same utility the
consumer picks the one with the lowest total storage. Note that the total storage cost
incurred by the consumer is $\sum_{t=1}^{T-1} \mathcal{C}(S_t)$.

Let $\Delta\mathcal{C}(q) = \mathcal{C}(q+1) - \mathcal{C}(q)$. Thus, $\Delta\mathcal{C}(q)$ represents the additional
cost of storing one extra unit, given that $x-1$ items are already stored.
To conduct our analysis, we assign the storage costs incurred over the units as follows.
We order the items in increasing order of purchase time; we arbitrarily order any units purchased in the same period.
We now assume that units are consumed in a {\em first-in-first-out} (FIFO) fashion; given a choice, the buyer consumes
the unit purchased earliest.
Observe that the storage cost at period $t$ is
$$\mathcal{C}(S_t) = \sum_{\ell=1}^{S_t} \Delta\mathcal{C}(\ell-1)$$
As we have an ordering on the items, we may assign to the $\ell$th of these $S_t$ units the storage cost
$\Delta\mathcal{C}(\ell-1)$.

Assume that $Q(\mathcal{D}_2)>0$. Take an interval $\mathcal{I}$ consisting of the periods $\{i, i+1, \dots, i+k=\tau\}$, such that
$S_i=S_{\tau}=0$ and $S_t>0$ for all $i<t<\tau$. By adding a dummy time period $0$ with
$S_0=0$, such an interval must exist as $S_T=0$.
Next, let $\alpha$ denote the last period in the $\mathcal{I}\setminus \{\tau\}$ at which the consumer
bought an item. Thus $q_{\alpha}>0$ and $q_{t}=0$ for all $\alpha < t < \tau$.

Because $S_{\tau-1}>0$ and $S_{\tau}=0$ we have that consumption $x_{\tau}$ at time $\tau$ is positive
and strictly exceeds the number of items $q_{\tau}$ purchased in that period.
Moreover, we claim that every item consumed at period $\tau$ was bought in period $\alpha$. Suppose not.
Then by the FIFO policy every item purchased at time $\alpha$ is consumed in period $\tau$, and at
least one item purchased in period $\beta$ is consumed in period $\tau$, where
$\beta < \alpha$ is the last period before $\alpha$ at which there was a purchase.
Instead of purchasing that item at time $\beta$ for consumption at time $\tau$, the consumer could
have purchased it at time $\alpha$. By viewing this as the last of the resultant $S_{\alpha}+1$
items stored at time $\alpha$, we see that this does not affect the
storage cost of any other item bought in any time period.
Thus, by consumer rationality, the price paid at time $\beta$ plus the storage
cost for the item over the interval $[\beta, \alpha]$ must be less than the price announced at period $\alpha$.
Specifically
$$ p_{\beta} + \Delta\mathcal{C}(S_{\beta}-1) + \Delta\mathcal{C}(S_{\beta+1}-1)
+ \hdots +  \Delta\mathcal{C}(S_{\alpha-1}-1) \ <\  p_{\alpha}$$
But then, by concavity of the storage costs, it is better for the consumer to purchase
all of the items currently bought in period $\alpha$,
in period $\beta$ instead. Thus, as claimed, every item consumed in period $\tau$
was bought in period $\alpha$.

Now suppose the retailer modifies her price commitment strategy by changing the price at period $\tau$ to
$$p'_{\tau} = p_{\alpha} + \sum_{j=\alpha}^{\tau-1} \Delta\mathcal{C}(|S_j|-1).$$
Call the resulting price-commitment strategy for the retailer $\mathcal{D}_1'$. Faced now with $\mathcal{D}_1'$,
we claim the consumer utility cannot decrease if she purchases as under $\mathcal{D}_1$ except
that  she buys $S_{\tau-1}$ items less in period $\alpha$ and $S_{\tau-1}$ items more in period $\tau$.
Under $\mathcal{S}$,  the price and storage costs associated with the last item bought at
period $\alpha$ is
$$ p_{\alpha} + \sum_{j=\alpha}^{\tau-1}\Delta\mathcal{C}(|S_j|-1)$$
But this is exactly the price $p'_\tau$. By FIFO and the assignment rule for storage costs,
postponing the purchase of that item
until time $\tau$ will not alter the storage costs assigned to any other item (note that storage
quantities after period $\tau$ remain unchanged).
Thus the consumer is indifferent to postponing the purchase of the item until time $\tau$
when the retailer uses $\mathcal{D}_1'$. But by the concavity of the storage costs, a similar
argument implies that
the consumer may postpone the
purchase of every item bought in period $\alpha$ (that is, $S_{\tau-1}$) until period $\tau$ without a reduction in utility
under $\mathcal{S}'$

In addition, we claim that the consumer cannot obtain a higher utility using another strategy against the new prices. If so,
an improvement was also possible against the old prices, a contradiction. To see this,
recall that the only storage modification induced by the retailer change of strategy from $\mathcal{D}_1$
to $\mathcal{D}_1'$ was a reduction in storage from period $\alpha$ to period $\tau$
of $S_{\tau-1}$. By FIFO and the assignment rule for storage costs, this storage reduction can only
contribute to an increase in the storage costs for a potential new item bought before period $\tau$, and for
items already bought before period $\tau$ the storage costs remain the same.
On the other hand, assume that it is now beneficial to buy at period $\tau$, for price $p'_{\tau}$, an
addition item or items for consumption at
time $\tau$ or later. But then, by concavity, it was also beneficial to do so under the old prices. Thus,
the postponement strategy is
optimal against the new prices.

To conclude, observe that the new prices yield a higher profit for the retailer than the old prices as $p'_{\tau} > p_{\alpha}$.
This contradicts the optimality of $\mathcal{D}_1$. Thus, there is an optimal solution without storage.
\end{proof}

\subsection{The Multi-Buyer Case.}\label{sec:concave-multi}
Unlike in the single-buyer case, Theorem \ref{lemma_buy_on_time} cannot be extended to the
 more general setting of concave costs when there are multiple buyers.

\begin{theorem}\label{theorem_multi_buyer_concave}
There are instances of the multi-buyer setting where storage costs are concave in which every
revenue maximising price commitment strategy induces some storage.
\end{theorem}
\begin{proof}
We provide a simple example with $T=3$ time periods. Let the storage cost
function $\mathcal{C}: \mathbb{N} \to \mathbb{R}_{\geq 0}$ be defined as:
%

$$ \mathcal{C}(q) := \left\{
\begin{array}{ll}
1.5 & \textrm{ if } q=1 \\
1.5 + \epsilon & \textrm{ if } q = 2
\end{array}\right. $$
Here $\epsilon$ is a tiny positive number.
Now let there be $N = n_1 + n_2 + 1$ buyers where $n_1$ and $n_2$ are very large numbers.
The first $n_1$ consumers have a value of 1 in the first period and zero in later periods.
The second $n_2$ consumers have zero value in the first two periods and a value of $4$ in the last period.
Finally the last consumer has a value of 2.75 in the second period, a value of 3 in the third
period and no value in the first period.

One strategy for the retailer is to set $p_1=1, p_2 = \infty , p_3 = 4$. In that case, $n_1$ consumers
would buy and consume at $t=1$, $n_2$ consumers would buy and consumer at $t=3$, and the remaining
consumer would buy at $t=1$ two units incurring in a storage cost of $1.5 + 1.5 + \epsilon = 3 + \epsilon$.
This pricing strategy yields a profit for the retailer of $(n_1 + 2)\cdot 1  + n_2 \cdot 4$. Observe that is
almost optimal in the sense that all except one consumer pay their total valuation.

Suppose there exists another pricing strategy $(p'_1, p'_2, p'_3)$ that is at least as profitable as the one
mentioned above, but that induces no storage. Since $n_1$ and $n_2$ are very large numbers, any such
price commitment strategy requires that the first $n_1$ buyers pay a price of $1$ and the $n_2$ consumers
pay a price of $4$. Thus $p'_1=1$ and $p'_3=4$.

If the remaining consumer buys at period 1, she obtains an utility of
$2.75 + 3 - 2\cdot 1 - (1.5+\epsilon) - 1.5 = 0.75 - \epsilon$. Thus in order for the
consumer to buy at period 2 (and not use storage) at a price $p'_2$ her utility has to be
at least equal to $0.75 - \epsilon$. If the consumer buys one unit at $t=2$ we have that
$2.75 -  p'_2  \geq 0.75 -\epsilon $. Therefore, $p'_2 \leq 2 + \epsilon$.
On the other hand, if the consumer buys two units at $t=2$ we have that
$2.75 + 3 - 2 \cdot p'_2 - 1.5 \geq 0.75 -\epsilon$.
This implies that $p'_2 \leq 1.875 + \epsilon/2$.
Observe now than in any of these cases, $p'_2$ needs to be so low that the $n_2$ consumers
would always prefer to buy at $t=2$ and store rather than to buy at $t=3$. As a result, we proved
that in this game every optimal price commitment strategy induces some storage.
\end{proof}

\section{Conclusions and future research}
In this paper we studied the dynamic pricing problem of a monopolistic retailer who sells indivisible storable goods to strategic consumers under a finite time horizon. The case where storable goods are arbitrarily divisible has been studied by \cite{dudine2006storable}. We focused in the understanding of two major pricing mechanisms: (1) \emph{preannounced pricing}, where the retailer announces at the beginning of the game the price to be charged at each time period; (2) \emph{contingent pricing}, the retailer announces prices sequentially at each time period, which means that the price announced is allowed to be contingent to game history.

First, we showed that computing the optimal preannounced pricing policy can be done in polynomial time (specifically in $O(N^2T^3)$ time). This result is in line with divisible goods model where the authors have shown that finding the optimal preannounced pricing can be done also efficiently by solving a different dynamic program (see supplemental appendix B in \cite{dudine2006storable}).

We constructed a simple example to show that under the contingent pricing policy rather than the preannounced pricing mechanism, (i) prices can be lower, (ii) retailer revenues can be higher, and (iii) consumer surplus can be higher. Surprisingly, these three facts are in complete contrast to the case of a retailer selling divisible storable goods \citep{dudine2006storable}. Intuitively, under contingent pricing, the retailer can sometimes extract more revenue when the goods are indivisible by performing some price discrimination difficult to overcome for the consumer. When the retailer tries to perform a similar price discrimination under the divisible goods setting, the consumer can finely optimize the (continuous) quantities bought, making the price discrimination less effective. The example illustrates precisely this phenomenon. We then proceed to quantify how much more profitable can the contingent pricing mechanism be with respect to preannounced pricing. We showed that for a market with $N$ consumers, a contingent policy can produce a multiplicative factor of $\Omega(\log N)$ more revenues than a preannounced policy. This is surprising, since to our knowledge, all previous monopoly models with a finite number of periods have the property that the
contingent pricing mechanism cannot significantly outperform the preannounced pricing mechanism.
Since preannounced pricing strategies are typically easy to implement and popular with consumers, this means there is little practical benefit for a retailer to explore the use of more complex contingent pricing mechanisms. Our results, however, provide a setting in which it may indeed be beneficial for the retailer to explore such options.

This paper leaves interesting questions for future research. One of them is to extend the preannounced pricing algorithm of this paper to more general settings such as when storage costs are concave or convex. Given that under concave costs, consumers might store goods under preannounced pricing, the construction of such algorithm might not be simple. While our results show that there are cases where the retailer revenue is much larger using contingent pricing rather than preannounced pricing, it is unclear whether the opposite is also true. That is, are there cases where the revenue obtained using a preannounced pricing is arbitrarily larger than by using a contingent pricing mechanism? Another research direction is to characterize the (subgame perfect) Nash equilibrium of the price contingent mechanism. \citet{berbeglia2014} characterized the subgame perfect equilibrium for a durable good pricing problem (which is special case of the problem studied here) and showed that such equilibrium can be computed in polynomial time. Extending such characterization or providing a hardness result for the problem studied in this paper is an interesting challenge. Finally, it would be interesting to extend the model studied in this paper by incorporating additional constraints appearing often in practice such as, for example, imposing an extra fixed cost every time a consumer makes an order (i.e. ordering cost).

\section*{Acknowledgements}
The authors would like to thank to Gustavo Vulcano and Jun Xiao for helpful discussions.

\bibliography{references_duropoly}

\end{document}